\documentclass[11pt]{article}%font size
\usepackage{graphicx} % Required for inserting images
\usepackage{amsfonts}
\usepackage[utf8]{inputenc}
\usepackage[english]{babel}
\usepackage[margin=1in]{geometry}%margin
 \usepackage{setspace}
\usepackage{parskip}
\usepackage{enumitem}
\usepackage{authblk} 
\usepackage{mathtools}
\newcommand{\comment}[1]{} %define a new command that effectively does nothing with the input
\usepackage{blindtext}
\usepackage{amssymb}
\usepackage{hyperref}
\usepackage[dvipsnames]{xcolor}
\usepackage{xcolor}

\definecolor{darkgreen}{RGB}{0, 150, 0} % Custom darker green
\hypersetup{
     colorlinks=true,
     linkcolor=teal,
     filecolor=blue,
     citecolor = Maroon,      
     urlcolor=RoyalPurple,
     }
\usepackage{textcomp}
\usepackage{amsmath,amsfonts,amssymb,amsthm}
\usepackage{mathtools}
\usepackage{commath}
\usepackage[sc,osf]{mathpazo}
\usepackage{amsthm}
\usepackage{amsmath}

\newtheorem{theorem}{Theorem}[section]
\newtheorem{Lemma}{Lemma}[section]
\def\P{\mathbf{P}}
\def\oP{\overline{\mathbf{P}}}
\def\E{\mathbb{E}}
\def\am{\mathbf{M}}
\def\aom{\overline{\mathbf{M}}}

\def\ax{\mathbf{X}}
\def\aox{\overline{\mathbf{X}}}
\def\al{\lambda}
\def\olam{\overline{\lambda}}
\def\Lam{\mathbf{\Lambda}}
\def\oLam{\overline{\mathbf{\Lambda}}}
\def\Z{\mathbf{Z}}
\newcommand{\eref}[1]{~(\ref{#1})}%{eq.~(\ref{#1})}
\newcommand{\N}[2]{{\mathbf{N}}^{#1}_{#2}}
\newcommand{\nt}[2]{{\mathbf{N}}^{(#1,T)}_{#2}}
\newcommand{\mt}[2]{{\mathbf{M}}^{(#1,T)}_{#2}}
\newcommand{\M}[2]{{\mathbf{M}}^{#1}_{#2}}

\newcommand{\xt}[2]{{\mathbf{X}}^{(#1,T)}_{#2}}
\newcommand{\lam}[2]{{\lambda}^{#1}_{#2}}
\newcommand{\lt}[2]{{\lambda}^{(#1,T)}_{#2}}

\newcommand{\blt}[2]{{\mathbf{\Lambda}}^{(#1,T)}_{#2}}

\def\coloneqq{:=}
\def\P{\mathbf{P}}
\def\E{\mathbb{E}}
\def\am{\mathbf{M}}

\def\ax{\mathbf{X}}
\def\aox{\overline{\mathbf{X}}}
\def\al{\lambda}
\def\olam{\overline{\lambda}}
\def\Lam{\mathbf{\Lambda}}
\def\oLam{\overline{\mathbf{\Lambda}}}
\def\Z{\mathbf{Z}}
\def\1{\mathbf{1}}
\def\half{\frac{1}{2}}

\def\f{\frac}

\def\intot{\int_0^t}
\begin{document}
\title{\bf Asymmetric super-Heston-rough volatility model with Zumbach effect as scaling limit of quadratic Hawkes processes}

\author{
Priyanka Chudasama\thanks{\href{mailto:priyankac@iisc.ac.in}{priyankac@iisc.ac.in}} \hspace{1pt} and
Srikanth Krishnan Iyer\thanks{Corresponding author: \href{mailto:skiyer@iisc.ac.in}{skiyer@iisc.ac.in}}\\
\small Department of Mathematics, Indian Institute of Science, Bangalore, India
}

\date{} % This removes the date

\maketitle

\begin{center}
    {\bf Abstract}
\end{center}
Hawkes processes were first introduced to obtain microscopic models for the rough volatility observed in asset prices. Scaling limits of such processes leads to the rough-Heston model that describes the macroscopic behavior. Blanc et al. (2017) show that Time-reversal asymmetry (TRA) or the Zumbach effect can be modeled using Quadratic Hawkes (QHawkes) processes. Dandapani et al. (2021) obtain a super-rough-Heston model as scaling limit of QHawkes processes in the case where the impact of buying and selling actions are symmetric. To model asymmetry in buying and selling actions, we propose a bivariate QHawkes process and derive a super-rough-Heston model as scaling limits for the price process in the stable and near-unstable regimes that preserves TRA. A new feature of the limiting process in the near-unstable regime is that the two driving Brownian motions exhibit a stochastic covariation that depends on the spot volatility.

\vspace{0.2in}
% AMS 2000 subject classifications: 60F05, 60F17, 60G55, 62P05. \\

\section{Introduction}
\label{s1}

``The hunt for a “perfect” statistical model of financial markets is
still going on,'' wrote the authors in \cite{finan}. The goal is to derive models for macroscopic behavior of financial asset prices that arise as
scaling limits of stochastic models that incorporate micro-structural features  observed in time series of price processes (see \cite{bacry,finan,omar, aditi, rough_fractional, unstable}).
Properties of interest include volatility clustering, leverage effect, fat (power-law) tails of the return distribution, etc. Volatility clustering is the tendency in the price process where large changes are usually followed by large changes of either sign while small changes usually follow small changes. The leverage effect refers to the negative correlation between return and volatility. That is, the volatility drops as asset prices increase and vice-versa. Other stylized facts observed in price movements are that the markets are highly endogenous, there is an absence of statistical arbitrage and price movements are statistically asymmetric.

In 2013, Bacry et al. \cite{bacry} proposed using Hawkes processes to study the macroscopic behavior of price processes. A Hawkes process $(N_t: t \geq 0)$ is a self-exciting point process that can be fully determined using the intensity (event arrival rate) function,
$$\lambda_t \coloneqq \mu+ \sum\limits_{t_i < t} \phi(t-t_i) = \mu+ \int_0^t \phi(t-s) \, dN_s.$$
Here, $\mu$ is the baseline intensity and $\phi$ is a nonnegative function that satisfies $\norm{\phi}_1 \coloneqq \int_0^\infty \phi(s)ds < 1$. The later condition guarantees stability, that is, the existence of a stationary distribution for the intensity process. $0 < t_1 < t_2 < \cdots$ denote the jump times of the process $N(t)$. Each arrival results in a jump in the intensity and the impact of a jump on the intensity decays according to the kernel $\phi$. This feedback from the process to its intensity is referred to as the self-exciting property and leads to clustering of arrivals.

Price processes can be modeled using Hawkes process in two ways. One approach,  as introduced in \cite{bacry}, is to use a bivariate Hawkes process $(\N{1}{t},\N{2}{t}:t \geq 0)$ representing upward and downward movements of the price process. The price process is then given by $\P_t \coloneqq \N{1}{t}-\N{2}{t}$. 

Alternately, one can use a univariate Hawkes process $N(t)$ to model jump events in  the price process where the size of jumps in price takes values $+1$ or $-1$ with arbitrary probability independent of everything else. Thus,
\begin{equation}
    P_t = \sum_{i=1}^{N(t)} \xi_i,
\label{e0}
\end{equation}
where $\xi_i \stackrel{iid}{\sim} \pm 1$ with probability $\mathbb{P}(\xi_1 = 1) = p$. 

In the first approach, the intensity function of the bivariate process satisfies
\[ \left( \begin{array}{c} \lambda_t^1 \\ \lambda_t^2 \end{array} \right) =
\left( \begin{array}{c} \mu^1 \\ \mu^2 \end{array} \right) + \intot \Phi(t-s) \left( \begin{array}{c} dN_s^1 \\ dN_s^2 \end{array} \right).\]
where
\begin{equation}
    \Phi = \left( \begin{array}{lr} \phi_1 & c \phi_2 \\ \phi_2 & \phi_1 + (c - 1) \phi_2 \end{array} \right).
    \label{e0a}
\end{equation}
In this case, the condition under which the intensity process has a stationary distribution, also known as the stability condition, is given by $\rho(K) < 1$ where $\rho(K)$ is the spectral norm of the matrix
\[ K = \int_0^{\infty} \Phi(t) \, dt. \]
Scaling limits of such processes are Gaussian processes with fixed volatility. The precise scaling procedure will be described later. To motivate the model we wish to study, we discuss the following important observations made in \cite{omar}. The only way to obtain stochastic volatility in the above set-up is to consider a scaling regime where $\rho(K)$ approaches one at an appropriate rate. This is referred to as the near unstable regime. Secondly, to ensure absence of arbitrage one must have $\mu_1 = \mu_2$ and the row sums of $\Phi$ must be identical (which explains the particular choice of the matrix in \eref{e0a}). Finally, negative correlations between price returns and volatility increments occur in this model only if $c > 1$. Under these conditions one obtains the Heston model in the scaling limit \cite{omar}, that is, the price process satisfies,
\begin{equation}
    dP_t = \f{1}{1-(\parallel \phi_1 \parallel_1 - \parallel \phi_2 \parallel_1)} \sqrt{\f{2}{1+c}} \sqrt{V_t} \, dW_t,
    \label{e1}
\end{equation}
where the process $V$ follows the SDE
\begin{equation}
    dV_t = \kappa(v_0 - V_t) dt + \eta \sqrt{V_t} \, dB_t.
    \label{e2}
\end{equation}
The covariation between the standard Brownian motions $W, B$ is given by
\begin{equation*}
    d\langle W,B \rangle_t = \f{1 - c}{\sqrt{2(1+c^2)}} \; dt,
    %\label{e3}
\end{equation*}
which is negative only if $c > 1$.

Yet another stylized fact observed in the markets is the presence of long range correlations. Such a phenomenon could arise, for instance, due to large meta orders getting split by algorithms to avoid impact of orders on market prices. In \cite{omar} it is shown that long range dependence can be obtained in the above set-up of a bivariate Hawkes process by imposing the following tail-decay condition on the largest eigenvalue of the excitation kernel $\Phi$, that is,
\[ \int_t^{\infty} (\phi_1+ c \phi_2)(s) \, ds \sim C t^{-\alpha}, \qquad \alpha \in \left(\half, 1\right). \]
The resulting scaling limit is the rough-Heston model
where the price process satisfies\eref{e1} but with \eref{e2} replaced by
\begin{equation*}
    V_t = V_0 + \intot (t-s)^{\alpha - 1} \left(\kappa(1+c - V_s) \, ds + \eta \sqrt{V_s} \, dB_s \right),
    %\label{e4}
\end{equation*}
for suitable constants $\kappa, \eta.$ The roughness here is a consequence of the singular factor in the integrand.

Another important but less discussed aspect of the price process, as observed by Zumbach in \cite{time_reversal}, is the statistical symmetry at high frequency on the macro scale when past and future are interchanged. This is referred to as time-reversal asymmetry (TRA) or the Zumbach effect. This asymmetry manifests itself in two ways. One is that the past returns negatively affect future volatility but not other way around. Secondly, past large scale realized volatility is more correlated with future small scale realized volatility than vice-versa.

In 2015, Blanc et al. \cite{blanc} proposed using a Quadratic Hawkes (QHawkes) process to model TRA. It turns out that a scaling limit of this model exhibits rough volatility behavior even in the stable regime. A QHawkes process is a self-exciting process similar to the Hawkes process with an additional quadratic term in its intensity function. Here one follows the second approach to modeling a price process using a univariate Hawkes process as described in \eref{e0} with $p=\half$. (we will return to this point later). The QHawkes counting process $(\N{}{t})_{T > 0}$ has intensity
\begin{equation}
    \lambda_t = \mu + \int_0^t \phi(t-s) \, d\N{}{s} + \Z_t^2, \text{ with } \Z_t \coloneqq \int_0^tk(t-s) \, d\P_s.
    \label{e5}
\end{equation}
The quadratic term in the intensity process leads to TRA. Increases in the intensity of the arrival of future events occur whenever price movements of the same sign are seen in succession. In other words, price momentum leads to an increase in volatility. The QHawkes price process exhibits volatility clustering, leverage effect, and fat tails of return distribution on the microscale. It was shown in \cite{aditi} that the scaling limit of such processes yields the super-Heston-rough process. The limiting price process $P_t$ satisfies $dP_t = \sqrt{V_t}\;dW_t$, where $W$ is a standard Brownian motion. The volatility process takes the form
\begin{equation*}
    V_t = V_0 + \frac{1}{\alpha}\int_0^t(t-s)^{\alpha-1} \lambda (\theta - V_s) \, ds + \frac{\lambda \nu}{\alpha}(\int_0^t(t-s)^{\alpha-1} \sqrt{V_s} \,dB_s)^{\beta},
    %\label{e6}
\end{equation*}
where $B_t$ is a standard Brownian motion independent of $W$, $\beta = 2$ and $\alpha \in(\frac{1}{2},1)$. Note that $\beta =1$ yields the rough-Heston model. A super-Heston model would have powers higher than the classical square-root term in the spot volatility evolution, represented here by the fact that $\beta > 1$.  

The choice of $p = \half$ yields a martingale inside the quadratic term in \eref{e5} which plays an important role in deriving the scaling limit.
However, this leads to the restriction that the impact of buying and selling actions is symmetric. This is the main issue we address in this paper. We shall keep notations closely aligned with that in \cite{aditi} for ease of comparison.

To model asymmetry in the impact of buy and sell actions, we go back to the first modeling approach of using a bivariate process, but do so with QHawkes processes instead. The price process $\P_t \coloneqq \N{1}{t}-\N{2}{t}$ where $\N{1}{t}$ and $\N{2}{t}$ are two QHawkes processes with respective intensities given by
\begin{align}
    \lam{1}{t}\coloneqq\mu_1+\int_0^t\phi_1(t-s) \,d\N{1}{s}+\int_0^t\phi_2(t-s) \,d\N{2}{s} + \left(\int_0^tk_1(t-s) \,d\M{1}{s}-\int_0^tk_2(t-s) \,d\M{2}{s}\right)^2 \nonumber\\
     \lam{2}{t}\coloneqq\mu_2+\int_0^t\phi_2(t-s) \,d\N{1}{s} +\int_0^t\phi_1(t-s) \,d\N{2}{s}+ \left(\int_0^tk_2(t-s) \,d\M{1}{s}-\int_0^tk_1(t-s) \,d\M{2}{s}\right)^2,
     \label{e7}
\end{align}
where $\M{1}{t} = \N{1}{t} - \int_0^t \lam{1}{s} \,ds$ and $\M{2}{t} = \N{2}{t} - \int_0^t \lam{2}{s} \,ds$ are martingales. If $k_1 > k_2$ then a larger-than-expected upward movement in prices will result in an increase in $\lam{1}{}$ and likewise for the downward movement.

We show that the scaling limit of the above price process is a super-Heston-rough volatility model that preserves the Zumbach effect at the macroscale and the asymmetry in price movements at microscopic level. In comparison with \eref{e0a}, note that we are setting $c=1$ in the linear term. We show that it is possible to obtain a negative correlation between returns and volatility increments without having to take $c > 1$ in the nearly-unstable regime. In fact we show that the cross variation between the two Brownian motions is itself stochastic. Depending on the parameters and the price trends, this could take both positive and negative values, thus allowing for a broader range of possibilities.
%While allowing for this additional parameter would be desirable, it complicates the analysis.

\section{Stability of the Price Process}
\label{s2}

The structure specified in \eref{e7} lends itself to a nice simplification of the intensity function of the price process. The intensity function $\al_{t}\coloneqq\lam{1}{t}-\lam{2}{t}$ of the price process satisfies
\begin{align*}%\label{e8}
     \al_{t} = \mu+\int_0^t\phi(t-s) \, d\P_s + \left[\int_0^t(k_1+k_2)(t-s) \, d\am_{s}\right]\left[\int_0^t(k_1-k_2)(t-s) \, d\am_{s}\right],
\end{align*}
where $\mu\coloneqq\mu_1-\mu_2$ and $\phi\coloneqq\phi_1-\phi_2$. $\am_{t} \coloneqq \M{1}{t} -\M{2}{t}$ is a martingale as the counting processes $\N{1}{t}$ and $\N{2}{t}$ do not have simultaneous jumps. Since we shall assume that $k_1 > k_2$, the model can be simplified by taking $k_1 + k_2 \coloneqq \sqrt{\alpha_1}k$ and $k_1-k_2 \coloneqq \sqrt{\alpha_2}k$ where $\alpha_1, \alpha_2>0$. Setting $\alpha\coloneqq\sqrt{\alpha_1\alpha_2}$, $\al_{t}$ can be rewritten as
\begin{align}\label{e9}
     \al_{t}=\mu+\int_0^t\phi(t-s) \, d\P_s + \alpha\left(\int_0^tk(t-s) \, d\am_{s}\right)^2.
\end{align}
Thus the price process $\P_t$ has intensity given by \eref{e9} with associated martingale
\begin{align*}%\label{e10}
    \am_{t} = \P_t - \int_0^t \al_{s} \, ds.
     \end{align*}
We now examine conditions under which the price process is stable. This follows (see \cite{stable_hawkes, nonlinear_hawkes}) provided the mean intensity $\E(\lambda_t)$ converges. Taking expectations in \eref{e9} and using the fact that 
$$\langle\am\rangle_t = \langle M^1 \rangle_t + \langle M^2 \rangle_t = \langle N^1 \rangle_t + \langle N^2 \rangle_t = \int_0^t (\lam{1}{s}+\lam{2}{s}) \, ds,$$
we obtain
\begin{alignat*}{1} %\label{e10a}
    \E[\al_{t}] 
    & =\; \mu+\int_0^t\phi(t-s)\E[\al_{s}]ds + \alpha\E\left[ \left\langle\int_0^tk(t-s)d\am_{s}\right\rangle \right] \nonumber \\
     & = \; \mu+\int_0^t\phi(t-s)\E[\al_{s}^{}]ds + \alpha\E\left[\int_0^t(k(t-s))^2d\langle\am\rangle_s\right] \nonumber \\
    & =\; \mu + \int_0^t\phi(t-s) \E[\al_{s}^{}] \, ds +
    \alpha \E\left[\int_0^tk^2(t-s)(\lam{1}{s}+\lam{2}{s}) \, ds\right].
\end{alignat*}
Let $\olam_{t} = \lam{1}{t}+\lam{2}{t}$, $\overline{\mu}\coloneqq\mu_1+\mu_2$ and $\overline{\phi}\coloneqq\phi_1+\phi_2$. Proceeding as above we obtain
\begin{alignat}{1}\label{e11}
      \E[\olam_{t}] 
%      & = \; \overline{\mu}+\int_0^t\overline{\phi}(t-
% s)\E[\olam_{s}]ds + \frac{(\alpha_1+\alpha_2)}
%{2}\E\left[\int_0^tk^2(t-s)\olam_{s}ds\right]\\
& =\; \overline{\mu}+\int_0^t\overline{\phi}(t-s)\E[\olam_{s}] \, ds + \frac{(\alpha_1+\alpha_2)}{2}\E\left[ \left\langle\int_0^tk(t-s)d\am_{s}\right\rangle \right] \nonumber \\
     & = \; \overline{\mu}+\int_0^t\overline{\phi}(t-s)\E[\olam_{s}] \, ds + \frac{(\alpha_1+\alpha_2)}{2}\E\left[\int_0^t(k(t-s))^2d\langle\am\rangle_s\right] \nonumber \\
 & = \; \overline{\mu}+\int_0^t\left(\overline{\phi}(t-s) + \frac{(\alpha_1+\alpha_2)}{2}k^2(t-s)\right)\E[\olam_{s}] \, ds.
 \end{alignat}
Convergence of $\E[\olam_{t}]$ will follow by monotonicity provided we can show that it is bounded. This in turn will imply the convergence of $\E[\al_{t}]$ (again by boundedness and monotonicity of $\E[\lam{i}{t}]$).  From \eref{e11} it is clear that $\E[\olam_{s}] \leq \E[\olam_{t}]$ for any $s < t$. Using this on right-hand side of \eref{e11} and simplifying, we obtain
\begin{align*}%\label{e12}
    \E[\olam_{t}] \leq \frac{\overline{\mu}}{1-\left(\norm{\overline{\phi}}_1 + \frac{(\alpha_1+\alpha_2)}{2}\norm{k}_2^2\right)}.
\end{align*}
We thus obtain the following sufficient condition for stability.
\begin{align*}%\label%{e13}
    \norm{\overline{\phi}}_1 + \frac{(\alpha_1+\alpha_2)}{2}\norm{k}_2^2 < 1.
\end{align*}
This condition is also required to show existence of subsequential limits for the sequence of scaled processes which we shall describe next.

\section{Main Results}
\label{s3}

\subsection{The Stable Regime}
\label{s3a}

To obtain a description for the macroscopic behavior, we consider a sequence of processes indexed by a parameter $T>0$. The price processes $\P_t^{T} = \mathbf{N}^{(1,T)}_t-\mathbf{N}^{(2,T)}_t$, $t \in [0,T]$ has intensity 
\[ \al_{t}^T = \lt{1}{t} - \lt{2}{t},\]
where following \eref{e7} we have
\small{
\begin{alignat*}{1}
    \lt{1}{t}\coloneqq & \mu_1^T \!+\! \int_0^t\phi_1^T(t-s) \,d\nt{1}{s}+\int_0^t\phi_2^T(t-s) \,d\nt{2}{s} \!+\! \left(\int_0^t k_1^T(t-s) \,d\M{(1,T)}{s}-\int_0^tk_2^T(t-s) \,d\M{(2,T)}{s}\right)^2, \nonumber\\
     \lt{2}{t}\coloneqq & \mu_2^T \!+\! \int_0^t\phi_2^T(t-s) \,d\nt{1}{s} \!+\! \int_0^t\phi_1(^Tt-s) \,d\nt{2}{s} \!+\! \left(\int_0^t k_2^T(t-s) \,d\M{(1,T)}{s}-\int_0^tk_1^T(t-s) \,d\M{(2,T)}{s}\right)^2.
     %\label{e7new}
\end{alignat*}}
Analogous to \eref{e9} the intensity of the price process $\al_{t}^T = \lt{1}{t} - \lt{2}{t}$ can then be written as
\begin{align}\label{e14}
     \al_{t}^T=\mu+\int_0^t\phi^T(t-s) \, d\P_s^T + \alpha \left(\int_0^tk^T(t-s) \, d\am_{s}^T \right)^2,
\end{align}
where $\phi^T \coloneqq \phi_1^T -\phi_2^T$ , $k_1^T + k_2^T \coloneqq \sqrt{\alpha_1}\;k^T$, $k_1^T - k_2^T \coloneqq \sqrt{\alpha_2}\;k^T$ with $\alpha_1, \alpha_2 > 0$, and $\alpha\coloneqq\sqrt{\alpha_1\alpha_2}$. $\am_{t}^T$ is the martingale $\P_t^T - \int_0^t \al_s^T \, ds$ and we have set $\mu^T \coloneqq \mu_1^T - \mu_2^T = \mu$. 
Rescaling $t \rightarrow tT$ and $s \rightarrow sT$ in \eref{e14} we obtain
\begin{alignat*}{1}
    \al_{tT}^T &=\; \mu +  \int_0^t\phi^T(T(t-s)) \, d\P_{Ts}^{T}+\alpha\left(\int_0^tk^T(T(t-s)) \, d\am_{Ts}^T\right)^2, \; t \in [0,1],
\end{alignat*}
which can be rewritten as
\begin{alignat}{1}\label{e15}
    \al_{tT}^T &=\;\mu +  \int_0^t\phi^T(T(t-s))T\al_{Ts}^T \; ds+\int_0^t\phi^T(T(t-s))\;d\am_{Ts}^T\nonumber\\
    & \; +\; \alpha\left(\int_0^tk^T(T(t-s))\;d\am_{Ts}^T\right)^2, \; t \in [0,1].
\end{alignat}
Similarly, the intensity $\olam_{t}^T = \lt{1}{t} + \lt{2}{t}$ of the processes $\oP^T_t \coloneqq \nt{1}{t} + \nt{2}{t}$ can be written as
\begin{align*}%\label{eo14}
     \olam_{t}^T=\overline{\mu}+\int_0^t\overline{\phi}^T(t-s) \, d\oP_s^T +\frac{\alpha_1 + \alpha_2}{2} \left(\int_0^tk^T(t-s) \, d\am_{s}^T \right)^2, \; t \in [0,T],
\end{align*}
where $\overline{\phi}^T \coloneqq \phi_1^T + \phi_2^T$ and we have set $\overline{\mu}^T \coloneqq \mu_1^T + \mu_2^T = \overline{\mu}$. 
Rescaling as above we get
\begin{alignat*}{1}%\label{eo15}
    \olam_{tT}^T &=\;\overline{\mu} +  \int_0^t\overline{\phi}^T(T(t-s))T\olam_{Ts}^T \; ds+\int_0^t\overline{\phi}^T(T(t-s))\;d\aom_{Ts}^T\nonumber\\
    & \; +\; \frac{\alpha_1 + \alpha_2}{2}\left(\int_0^tk^T(T(t-s))\;d\am_{Ts}^T\right)^2, \; t \in [0,1],
\end{alignat*}
where $\aom_{t}^T$ is the martingale $\oP_t^T - \int_0^t \olam_s^T$.

It is natural to try and show convergence of the sequence of processes $\frac{1}{\sqrt{T}}\am_{Tt}^T$. The appropriate scaling assumptions for the parameters given below are then evident. We shall add a technical condition that will be required to prove tightness.

%\addtocontents{toc}{\protect\setcounter{tocdepth}{1}}
\subsubsection{Assumptions}\label{assumpm2}
%\addtocontents{toc}{\protect\setcounter{tocdepth}{2}}
%----------------------------------------------------------------------
\renewcommand{\theenumi}{\roman{enumi}}%
\begin{enumerate}
 \item The sequence of kernels satisfy 
\begin{align*}%\label{aa1}
\phi^T(\cdot) = \phi\left(\frac{\cdot}{T}\right)\frac{\beta}{T} ,\;\;\;\; \overline{\phi}^T(\cdot) = \phi\left(\frac{\cdot}{T}\right)\frac{\overline{\beta}}{T} \;\;\text{ and } \;\; k^T(\cdot) = k\left(\frac{\cdot}{T}\right)\frac{1}{\sqrt{T}}
\end{align*}
for some constants $\beta, \overline\beta$ and kernels $\phi, k$ satisfying $\abs{\beta} < \overline\beta$ and 
\begin{equation*}
    0<\overline{\beta}\norm{\phi}_1+\frac{(\alpha_1+\alpha_2)}{2}\norm{k}_2^2<1.
    %\label{aa2}
\end{equation*}

\item  There exists $\eta ,\overline{\eta} \geq 0$ such that
\begin{align*}
     \norm{\phi}_1 + 2\alpha\norm{k}_2 \left(\eta + \overline{\mu}\right) < \frac{\eta}{\eta+\abs{\mu}}, \text{ and }%\label{a11}\\
     \norm{\phi}_1 + (\alpha_1+\alpha_2)\norm{k}_2\left(\overline{\eta} + \overline{\mu}\right) < \frac{\overline{\eta}}{\overline{\eta}+\overline{\mu}}.%\label{a12}
\end{align*}
These conditions will be required to prove the uniqueness of the limiting processes (see section \ref{uniq_m1}).

\item The function
$k \in L^{2+\epsilon}$ for some $\epsilon > 0$ and for any $0 \leq t \leq \hat{t} \leq 1, $
$$ \int_0^t \abs{k(\hat{t}-s)-k(t-s)}^2 \, ds < C\abs{\hat{t}-t}^r$$
for some $r, C > 0$. Further, for some $\eta \in (0,1).$
$$ \int_0^1\abs{k(t)}^2t^{-2\eta} \, dt +\int_0^1\int^1_0\frac{\abs{k(t)-k(s)}^2}{\abs{t-s}^{1+2\eta}} \, ds \, dt < \infty.$$
\end{enumerate}

\subsubsection{Scaling Limit in the Stable Regime}
\label{s3ab}

%Let $\overline{\P}_{tT}^{T} = \mathbf{N}^{1,T}_t  + \mathbf{N}^{2,T}_t$ and $\Lam_{t}^{T}=\int_0^t(\lt{1}{Ts} - \lt{2}{Ts})ds = \int_0^t \lambda_{Ts}^Tds$. 
Define the sequence of scaled processes $\ax_{t}^{T}=\frac{\P_{tT}^{T}}{T},$ $\am_t^{*T} = \frac{1}{\sqrt{T}}\am_{Tt}^T$, $t \in [0,1],$  $T>0.$
Consider the Mittag-Leffler function
\begin{equation}\label{mittag}
f^{\tilde{\alpha}, \sigma}(s)  \coloneqq \sum_{n=0}^{\infty}\frac{s^n}{\Gamma(\tilde{\alpha} n + \sigma)}.
\end{equation}
%
%where $\tilde{\alpha}, \sigma$ are complex parameters and $\Gamma(x)$ is the Gamma function. 
%If $\tilde{\alpha}, \sigma$ are real and positive,  $f^{\tilde{\alpha}, \sigma}$  becomes an entire function.
%Moreover, 
For $\tilde{\alpha} \in (\frac{1}{2}, 1)$ and $\sigma >0$, $f^{\tilde{\alpha}, \sigma}$ satisfies assumption \ref{assumpm2} (iii) for any $\epsilon \in  (o, \frac{2\tilde{\alpha} -1}{1 - \tilde{\alpha}})$, $\eta \in (0, \tilde{\alpha} - \frac{1}{2})$ and $r = 2\tilde{\alpha} -1$.
%
%\begin{align}\label{r11}
%        \ax_{t}^{T}=\frac{\P_{tT}^{T}}{T}, \qquad 
%    \overline{\ax}_{t}^{T}=\frac{\overline{\P}_{tT}^{T}}
%{T},\qquad \P_t^{*T}=\frac{\P_{tT}^{T}}{\sqrt{T}} \qquad
%    \text{ and } \qquad \am_{t}^{*T}\coloneqq \frac{1}
%{\sqrt{T}}\am_{Tt}^T.
%    \end{align}
%
We are now ready to state the first main result of this paper.

\begin{theorem}\label{theorem2}
% (i) and \ref{assumpm2} (ii)
Under the assumptions \ref{assumpm2}, the sequence $\left(\ax^T, \am^{*T} \right)_{T > 0}$ 
is C-tight for the Skorokhod topology on $[0,1]$, and as $T \to \infty$ converges to a process $(X,M^*)$ with the following properties.

1.\hspace{5pt}$X_t = \int_0^t V_s \, ds$, where $V$ is the \textbf{unique solution} of the equation 
\begin{equation*}
    V_t = \mu + \beta\int_0^t \phi(t-s)V_s \, ds + \alpha (Z^*_t)^2, \text{ with } Z_t^*= \int_0^t k(t-s)\, dM^*_s. %\label{lv1}
\end{equation*}
2.\hspace{5pt} There exists a Brownian motion $B$ such that 
\begin{equation*}
    M^*_t =\int_0^t \sqrt{\overline{V}_s}\;dB_s,%\label{lv3}
\end{equation*}
where $\overline{V}$ is 
%   the derivative of $\overline{X}$ and 
a solution of the equation 
\begin{equation*}
     \overline{V}_t = \overline{\mu} + \overline{\beta}\int_0^t \phi(t-s)\overline{V}_s \, ds + \frac{(\alpha_1+\alpha_2)}{2}(Z^*_t)^2.%\label{lv2}
\end{equation*}
3. For any $\epsilon > 0$, if $k = f^{\tilde{\alpha}, \sigma }$ with $\tilde{\alpha} \in (\frac{1}{2},1)$ and $\sigma >0$, then $V, \overline{V}$ have almost surely $\tilde{\alpha} -\frac{1}{2}-\epsilon$ H\H{o}lder regularity.
\end{theorem}

\subsection{The Near Unstable Regime}
\label{s3b}

A more interesting limit process arises when we consider the sequence of scaled processes that approaches the stability barrier.

Let 
\begin{equation}
a^T = \norm{\overline{\phi}^T}_1 + \frac{(\alpha_1+\alpha_2)}{2}\norm{k_T}_2^2. \label{a^T}
\end{equation} 
Recall that in deriving the scaling limit in the stable regime we had assumed that $\sup_{T>0} a^T < 1$. We now state the assumptions required to prove the scaling limit in the near unstable regime, that is $1 > a^T \to 1$ as $T \to \infty$.

\subsubsection{Assumptions}\label{assumpm3}
\renewcommand{\theenumi}{\roman{enumi}}%
\begin{enumerate}
 \item Fix $\phi \geq 0$ such that $\norm{\phi}_1 = 1$ and let 
\begin{align*}%\label{aa2n}
\phi^T(\cdot) = \phi\left(\frac{\cdot}{T}\right)\frac{\beta^T}{T} ,\;\;\;\; \overline{\phi}^T(\cdot) = \phi\left(\frac{\cdot}{T}\right)\frac{\overline{\beta}^T}{T},
\end{align*}
where $|\beta^T| = \frac{1}{1 + c_1(1-a^T)}$ and $\overline{\beta}^T = \frac{1}{1 + c_2(1-a^T)}$, where the constants $c_1, c_2$ satify $c_1 > c_2 > \frac{(\alpha_1+\alpha_2)}{2}$ and $1 > a^T \to 1$ as $T \to \infty$.

 \item There exists constants $K > 0$  and $\tilde{\alpha} \in (\half,1)$ such that
$$\tilde{\alpha} x^{\tilde{\alpha}} \int_x^{+\infty} \phi(s)ds \to K \text{ as } x \to +\infty.$$

 \item Let $\delta = K\frac{\Gamma(1-\tilde{\alpha})}{\tilde{\alpha}}$. There exists $\mu^*, \overline{\mu}^*, \sigma > 0$ and such that 
 \begin{equation*}%\label{aa3_nu}
     (1-a^T )T^ {\tilde{\alpha}} \to \sigma\delta, \; \mu^T T^{(1-\tilde{\alpha})} \to \frac{\mu^*}{\delta} \text{ and } \overline{\mu}^T T^{(1-\tilde{\alpha})} \to \frac{\overline{\mu}^*}{\delta} \text{ as } T \to \infty.
 \end{equation*}
 \item $k_T(\cdot) = k(\cdot/T) \sqrt{\frac{1-a^T}{T}}$, where $k$ is a non-negative continuously differentiable function such that $\norm{k}_2^2 = 1$.
 \end{enumerate}

{\bf Remark:} Observe that under the assumption \ref{assumpm3} (i) the right hand side of equation \eref{a^T} reduces to $$0< \frac{1}{1 + c_2(1-a^T)} + \frac{(\alpha_1+\alpha_2)}{2} (1- a^T) \to 1 \hbox{ as } T \to \infty.$$

\subsubsection{Scaling Limit in the Near Unstable Regime}
\label{s3bb}

Consider the sequence of scaled processes 
$\ax_{t}^{T}=\frac{1-a^T}{\overline{\mu}^T}\frac{\P_{tT}^{T}}{T},$ $\am_t^{*T} =\sqrt{\frac{1-a^T}{\overline{\mu}^T}}\frac{1}{\sqrt{T}}\am_{Tt}^T$, $t \in [0,1],$  $T>0.$
Let $f^{\tilde{\alpha}, \sigma}(s) = \sigma s^{\tilde{\alpha} -1} \sum_{n=0}^{\infty}\frac{(-\sigma s^{\tilde{\alpha}})^n}{\Gamma(\tilde{\alpha} n + \tilde{\alpha})}$ where the sum is the Mittag-Lefflar function as defined in \eref{mittag}. Further, let  $F^{\tilde{\alpha}, \sigma}(t) \coloneqq \int_0^t f^{\tilde{\alpha}, \sigma}(s)ds$.

\begin{theorem}\label{theorem3}
Under the assumptions \ref{assumpm3} (i)-(iv), the sequence $\left(\ax^T, \am^{*T}, t \in [0,1] \right)_{T > 0}$ 
is C-tight for the Skorokhod topology on $[0,1]$ and any subsequential limit process $\left( X_t,M_t^*, t \in [0,1] \right)$ will satisfy the following evolution equations.

\begin{equation} \label{X_t}
    X_t = \;\int_0^t \frac{1}{c_1}f^{\tilde{\alpha}, \sigma}(t-s)\frac{1}{\sqrt{\sigma\overline{\mu}^*}}\;M^*_{s}\;ds  + \int_0^t  \frac{1}{c_1}F^{\tilde{\alpha}, \sigma}(t-s)\left(\frac{\mu^*}{\overline{\mu}^*} + \alpha(Z^*_{s})^2\right)ds,
\end{equation}
with  $Z_t^*= \int_0^t k(t-s)\, dM^*_s$ and $\langle M^* \rangle = \overline{X}_t$ where,
\begin{equation} \label{X_bar_t}
\overline{X}_t = \;\int_0^t \frac{1}{c_2}f^{\tilde{\alpha}, \sigma}(t-s)\frac{1}{\sqrt{\sigma\overline{\mu}^*}}\;\overline{M}^*_{s}\;ds  + \int_0^t  \frac{1}{c_2}F^{\tilde{\alpha}, \sigma}(t-s)\left(1 + \frac{(\alpha_1+\alpha_2)}{2}(Z^*_{s})^2\right)ds.
\end{equation}
\textbf{Further, if $\tilde{\alpha} \in (\frac{1}{2}, 1)$, we have}

\hspace{0.5cm} 1. Upto an enlargement of the filtration, there exist Brownian motions $B^1,B^2$ such that 
\begin{equation}
    M^*_t =\int_0^t \sqrt{\overline{V}_s}\;dB^1_s, \qquad \overline{M}^*_t =\int_0^t \sqrt{\overline{V}_s}\;dB^2_s, \qquad \langle B^1, B^2\rangle_t = \frac{V_t}{\overline{V_t}} \label{lv4}
\end{equation}
\hspace{1.2cm} where, $V$  and $\overline{V_t}$  are solutions of 
\begin{alignat*}{1}
    V_t = &\;\int_0^t \frac{1}{c_1}f^{\tilde{\alpha}, \sigma}(t-s)\left[\frac{1}{\sqrt{\sigma\overline{\mu}^*}}\;dM^*_{s}  +\left(\frac{\mu^*}{\overline{\mu}^*} + \alpha(Z^*_{s})^2\right)ds\right],
\end{alignat*}

\begin{equation} 
    \overline{V}_t = \;\int_0^t \frac{1}{c_2}f^{\tilde{\alpha}, \sigma}(t-s)\left[\frac{1}{\sqrt{\sigma\overline{\mu}^*}}\;d\overline{M}^*_{s}  +\left(1 + \frac{(\alpha_1+\alpha_2)}{2}(Z^*_{s})^2\right)ds\right]. \label{lv5}
\end{equation}
\hspace{0.5cm} 2. The processes $(X_t,\overline{X_t})$ are almost surely differentiable with derivatives $(V_t, \overline{V_t})$.

\hspace{0.5cm} 3.  For any $\epsilon > 0$, $V$, $\overline{V}$ have almost surely $\tilde{\alpha}-\frac{1}{2}-\epsilon$ H\H{o}lder regularity.
\end{theorem}
% \begin{enumerate}[label=\Roman*., ref=\Roman*]
%     \item Outer item
%     \begin{enumerate}[label=\arabic*., ref=\arabic*]
%         \item Inner item
%         \item Another inner item   
%     \end{enumerate}
%     \item Another outer item
% \end{enumerate}
\section{Proofs}
In this section, we will prove Theorem \ref{theorem3}. The proof of tightness and the structure of the evolution equations in Theorem \ref{theorem2} is similar and is therefore omitted. The proof of uniqueness of the evolution equations uniqueness in Theorem \ref{theorem2} is given in Section~\ref{uniq_m1}. These proofs follow along similar lines to those in \cite{aditi}. We prove Theorem \ref{theorem3} into two parts. We will first prove C-tightness of the family 
$\left(\ax^T, \am_t^{*T}, t \in [0,1] \right)_{T > 0}$. This will ensure existence and path continuity of the subsequential limits. In the second part, we establish the dynamics of limit points to find the system of equations that any subsequential limit will satisfy.

Consider the process $\overline{\P}_t =\mathbf{N}^{(1,T)}_t + \mathbf{N}^{(2,T)}_t$, $t \in [0,T]$ with intensity 
\[ \overline{\al}_{t}^T = \lam{(1,T)}{t} + \lam{(2,T)}{t},\]
and the corresponding martingale $\overline{\am}_{t}^T = \overline{\P}_t - \int_0^t \overline{\al}_{s}^T \, ds.$
Define the scaled processes
\begin{equation*} %\label{eqn:XMbar*}
    \aox_{t}^{T} :=\frac{1-a^T}{\overline{\mu}^T}\frac{\overline{\P}_{tT}^{T}}{T}, \qquad \overline{\am}_t^{*T} =\sqrt{\frac{1-a^T}{\overline{\mu}^T}}\frac{1}{\sqrt{T}}\overline{\am}_{Tt}^T.
\end{equation*} 
We will also prove C-tightness of the family  $\left((\aox^T,\overline{\am}_t^{*T}), t \in [0,1] \right)_{T > 0}$ as their limits appear in \eref{lv4}, \eref{lv5}.

The scaled processes $\ax_{t}^{T}, \aox_{t}^{T}$ can be written as
\begin{equation*} %\label{pe1}
\ax_{t}^{T} = \xt{1}{t} - \xt{2}{t}, \qquad \aox_{t}^{T} = \xt{1}{t} + \xt{2}{t},
\end{equation*}
where
\begin{align*}%\label{x1_nu}
    \xt{i}{t}\coloneqq \frac{(1-a^T)}{\overline{\mu}^T}\frac{\nt{i}{tT}}{T}, \;\; i = 1,2.
    \end{align*}
The compensator of the process $\xt{i}{t}$ is given by
\begin{align}\label{Lami_nu}
    \blt{i}{t}\coloneqq \frac{(1-a^T)}{\overline{\mu}^T}\int_0^{t}\lt{i}{sT}\;ds, \;\; i = 1,2.
    \end{align}
Let
\begin{equation}
    \label{LamT}
    \Lam_t^{T} =  \blt{1}{t} - \blt{2}{t} = \frac{(1-a^T)}{\overline{\mu}^T}\int_0^{t}\al_{sT}^T\;ds, \qquad \oLam_t^{T} =  \blt{1}{t} + \blt{2}{t} = \frac{(1-a^T)}{\overline{\mu}^T}\int_0^{t}\overline{\al}_{sT}^T\;ds.
\end{equation}
%
%where $\blt{i}{t}$, $i=1,2$ are as defined in \eref{Lami_nu}.
%----------------------------------------------------------------------
\subsection{C-Tightness of the processes $\left((\ax^T, \aox^T,\am_t^{*T}, \overline{\am}_t^{*T}), t \in [0,1] \right)_{T > 0}$} \label{tightness_nu}
%---------------------------------------------------------------------
%\begin{proof}

%$\lt{i}{s}$
For the scaled martingale processes $\am_t^{*T}, \overline{\am}_t^{*T}$ we have
\begin{alignat*}{1}%\label{sharp6}
\left \langle \am^{*T} \right \rangle_t = \left \langle \overline{\am}^{*T} \right \rangle_t 
%& = \frac{1-a^T}{\overline{\mu}^T}\frac{1}{T} \left \langle \am^{T} \right \rangle_{Tt} \nonumber \\
& = \frac{1-a^T}{\overline{\mu}^T}\frac{1}{T} \left(\int_0^{tT} (\lt{1}{s}+\lt{2}{s}) \; ds\right) \nonumber \\
& = \blt{1}{t} + \blt{2}{t}.
\end{alignat*}
By Theorem VI-4.13 in [17], tightness of the family $\{(\am_t^{*T}, \overline{\am}_t^{*T}): t \in [0,1]\}_{T > 0}$ follows if we show that the corresponding family of sharp bracket processes is tight. Since linear combinations of tight processes are tight and joint tightness holds if the marginals are tight, it suffices to show C-tightness of the families of processes $\{\xt{i}{t}: t \in [0,1]\}_{T > 0}$, $\{\blt{i}{t}: t \in [0,1]\}_{T > 0}$, $i=1,2$.

First we shall show that these two families of processes are tight. Since $\xt{i}{t},\; \blt{i}{t}$ are right continuous increasing processes, it suffices to show tightness of the families of random variables $\xt{i}{1},\; \blt{i}{1}$. This in turn follows if we show that $\E[\xt{i}{1}]$ and $\E[\blt{i}{1}]$ are uniformly bounded in $T$. Clearly,
\begin{align*}%\label{2.1.1}
    \E[\xt{i}{1}] = \E[\blt{i}{1}] \leq 
    \E[\oLam_1^{T}].
\end{align*}
From (\ref{Lami_nu}), we have
\begin{align}\label{bound_Lam_bar_T}
\E[\oLam_1^{T}] =
\frac{(1-a^T)}{\overline{\mu}^T} \int_0^1 \E[\olam_{sT}^T] \; ds.
\end{align}
%
%where
%
%\[ \olam_{sT}^T = \lt{1}{sT} + \lt{2}{sT}.\]
%
Analogous to \eref{e11} we have
     \begin{alignat}{1}
     \E[\olam_{t}^T]
     % & =\; \overline{\mu}^T+\int_0^t\overline{\phi}^T(t-s)\E[\olam_{s}^T]ds + \frac{(\alpha_1+\alpha_2)}{2}\E\left[ \langle\int_0^tk^T(t-s)d\am_{s}^T\rangle \right] \nonumber \\
     % & = \; \overline{\mu}^T+\int_0^t\overline{\phi}^T(t-s)\E[\olam_{s}^{T}]ds + \frac{(\alpha_1+\alpha_2)}{2}\E\left[\int_0^t(k^T(t-s))^2d\langle\am^T\rangle_s\right] \nonumber \\
     & = \; \overline{\mu}^T+\int_0^t\left(\overline{\phi}^T(t-s) + \frac{(\alpha_1+\alpha_2)}{2}(k^T(t-s))^2\right)\E[\olam^T_{s}]ds.
     \label{101}
     \end{alignat}
Using $\olam^T_{s} \leq \olam^T_{t}$ for $s \leq t$ and then setting $t=1$ on the right in \eref{101} we obtain
\begin{align}\label{1a2}
    \E[\olam_{t}^T] \leq \frac{\overline{\mu}^T}{1-\left(\norm{\overline{\phi}^T}_1 + \frac{(\alpha_1+\alpha_2)}{2}\norm{k^T}_2^2\right)}
\end{align}
Using the assumptions \ref{assumpm3} to substitute for $\overline{\phi}^T,\, k^T$ in \eref{1a2} yields
\begin{align*}
\E[\olam_{t}^T] \leq \frac{\overline{\mu}^T}{1-\left(\overline{\beta}^T\norm{\phi}_1  + \frac{(\alpha_1+\alpha_2)}{2}(1-a^T)\norm{k}^2_2\right)}.
\end{align*}
Substituting the above bound in \eref{bound_Lam_bar_T} and using the assumptions \ref{assumpm3} (i), (iv), we obtain the bound
\begin{alignat}{1}
\E(\oLam_1^{T}) & \leq \frac{1-a^T}{1-\left(\overline{\beta}^T\norm{\phi}_1 + \frac{(\alpha_1+\alpha_2)}{2}(1-a^T)\norm{k}_2^2\right)} \nonumber \\
    & \leq \frac{1-a^T}{1-\left(\frac{1}{1+c_2(1-a^T)} + \frac{(\alpha_1+\alpha_2)}{2}(1-a^T)\right)} \nonumber \\
& \leq \frac{1}{\frac{c_2}{1 +c_2(1-a^T)} - \frac{(\alpha_1+\alpha_2)}{2}}
\to \frac{1}{c_2 - \frac{(\alpha_1+\alpha_2)}{2}} < \infty, \label{expect_unstable}
\end{alignat}
as $1-a^T \to 0$ as $T \to \infty$ and $c_2 > \frac{(\alpha_1+\alpha_2)}{2}$. This proves tightness. To verify that all convergent subsequences converge to processes with continuous paths, note that $\N{1}{t}$ and $\N{2}{t}$ are counting processes with jump size 1. Thus,
$$\abs{\Delta\ax_t^{(i,T)}} = \frac{1-a^T}{\overline{\mu}^T}\frac{\abs{\Delta\mathbf{N}_t^{(i,T)}}}{T} \leq \frac{1-a^T}{\overline{\mu}^TT} \leq \frac{(1-a^T)^2}{T^{1-\tilde{\alpha}}\overline{\mu}^T (1-a^T)T^{\tilde{\alpha}}} \to 0$$
by \ref{assumpm3} (iii) and the fact that $1-a^T \to 0$ as $T \to \infty$. Furthermore, $\abs{\Delta\Lam_1^T}=0$ as $\Lam_t^T$ is continuous on $[0,1]$. Therefore, by proposition VI-3.35 (from (iii) $\Rightarrow$ (i)) in \cite{jacod} we have that  $(\xt{i}{t})_{T > 0}$ and $(\blt{i}{t})_{T > 0}$, $i \in \{1,2\}$ are C-tight for the Skorokhod topology on $[0,1]$.
% \end{proof}

\subsection{Dynamics of the limit points}
In this section we find limit points for the family of processes $\left((\ax^T,\overline{\ax}^T, \am_t^{*T}, \overline{\am}_t^{*T}), t \in [0,1] \right)_{T > 0}$. 
By the following Lemma, it is enough to study the dynamics of limit points for the family of processes $\left( \Lam_t^{T} ,\oLam_t^{T},\am_t^{*T},\overline{\am}_t^{*T}, t \in [0,1] \right)_{T > 0}$.  

\begin{Lemma}\label{l10}
The sequence of martingales $\left(\ax_t^T - \Lam_t^{T}, t \in [0,1] \right)_{T > 0}$ and $\left(\aox^T - \oLam_t^{T}, t \in [0,1]  \right)_{T > 0}$ converge uniformly to $0$  in probability.
\end{Lemma}

\begin{proof} Observe that 
$\E\left[\sup\limits_{t \in [0,1]} \left(\aox_{t}^T - \oLam_{t}^T\right)^2\right]$ and $\E\left[\sup\limits_{t \in [0,1]} \left(\ax_{t}^T - \Lam_{t}^T\right)^2\right]$ are respectively bounded by
\[
\E\left[\sup\limits_{t \in [0,1]} \left((\xt{1}{t} - \blt{1}{t}) \pm (\xt{1}{t} - \blt{2}{t})\right)^2\right]
  \leq  2 \sum_{i=1}^2 \E\left[\sup\limits_{t \in [0,1]} \left(\xt{i}{t} - \blt{i}{t}\right)^2\right].
 \]
For $ i \in \{1,2\},$ we have 
\begin{alignat}{1}\label{112}
\E\left[\sup\limits_{t \in [0,1]} \left(\xt{i}{t} - \blt{i}{t}\right)^2\right]&= \E\left[\sup\limits_{t \in [0,1]} \left(\frac{\nt{i}{tT}}{T} - \int_0^t\lt{i}{sT}ds\right)^2\right]\nonumber\\
& = \E\left[\sup\limits_{t \in [0,1]} \left(\frac{\nt{i}{tT}}{T} - \frac{1}{T}\int_0^{tT}\lt{i}{s}ds\right)^2\right]\nonumber\\
& = \frac{1}{T^2}\E\left[\sup\limits_{t \in [0,1]} \left({\nt{i}{tT}} - \int_0^{tT}\lt{i}{s}ds\right)^2\right] \nonumber \\
& \leq \frac{4}{T^2} \; \E\left[\left(\mt{i}{T}\right)^2\right]
 =  \frac{4}{T^2}\;\E\left(\left[\mt{i}{}\right]_T \right),
\end{alignat} 
where the second-last inequality follows by the Doob's martingale inequality. Further,
\begin{align} \label{116}
\E\left(\left[\mt{i}{}\right]_T \right) = \E\left(\nt{i}{T}\right) = \E\left(\int_0^T\lt{i}{s}ds\right) = \E\left(T\int_0^1\lt{i}{Ts}ds\right) = E\left(T\blt{i}{1}\right) \leq \E(T\oLam_{1}^T).
\end{align} 
Substituting the bounds from \eref{expect_unstable},\eref{116} in \eref{112} we get 
\begin{align*}%\label{117}
\E\left[\sup\limits_{t \in [0,1]} \left(\xt{i}{t} - \blt{i}{t}\right)^2\right] \leq \frac{4}{T}\;\E(\oLam_{1}^T)
 \leq  \frac{4}{T}\frac{1}{\frac{c_2}{1 +c_2(1-a^T)} - \frac{(\alpha_1+\alpha_2)}{2}} \to 0.
\end{align*} 
as $T \rightarrow  \infty$ for $i \in \{1,2\}$ since $1-a^T \rightarrow 0 $ and $c_2 > \frac{(\alpha_1+\alpha_2)}{2}$.
\end{proof}

Let $(X,\overline{X},M^*, \overline{M}^* )$ be a limit point of $\left(\Lam^{T} ,\oLam^{T},\am^{*T},\overline{\am}^{*T}\right)_{T > 0}.$ By C-tightness, the process  $(X,M^*, \overline{X},\overline{M}^*)$ is continuous. By the Skhorohod representation theorem we may suppose that the convergence is almost surely uniform on $[0,1]$.

To derive the limit points of $( \Lam^{T} ,\oLam^T)$ as given in (\ref{LamT}) we begin by analysing the intensities $(\al^T, \olam^T)$. From \eref{e15} (with $\mu$ replaced by $\mu^T$) the intensity of the price process $P_{tT}^T$ satisfies
\begin{alignat}{1}
  \al_{tT}^T &=\;\mu^T +  \int_0^t\phi^T(T(t-s))T\al_{Ts}^Tds+\int_0^t\phi^T(T(t-s))\;d\am_{Ts}^T \; +\; \alpha\left[\int_0^tk^T(T(t-s))\;d\am_{Ts}^T\right]^2 \nonumber\\
    & \; = \; \;\mu^T + \alpha\left(\Z_{tT}^T\right)^2 +  \int_0^t\phi^T(T(t-s))T\al_{Ts}^Tds+\int_0^t\phi^T(T(t-s))\;d\am_{Ts}^T , \label{e16}
\end{alignat}
where
\[ \Z_{tT}^T \coloneqq \int_0^{tT} k_T(Tt-s)\am_s^T = \int_0^{t} k_T(T(t-s))\am_{sT}^T.\]
Proposition 2.1 from \cite{unstable} now yields
\[ \al_{tT}^T =\;\mu^T + \alpha\left(\Z_{tT}^T\right)^2 +  \int_0^tT\psi^T(T(t-s))\left(\mu^T + \alpha\left(\Z_{tT}^T\right)^2 \right)ds + \int_0^t\psi^T(T(t-s))\;d\am_{Ts}^T, \]
where $\psi^T \coloneqq \sum\limits_{i\geq0}(\beta^T \phi^T)^{\circledast i}$, with $\circledast$ denoting the convolution operator.

Therefore, $\al_{t}^{*T} \coloneqq \;\frac{(1-a^T)}{\overline{\mu}^T}\al_t^T$ is given by
\begin{equation} \label{e17} \al_{tT}^{*T} =\;\frac{(1-a^T)\mu^T}{\overline{\mu}^T} + \frac{1-a^T}{\overline{\mu}^T} \alpha (\mathbf{Z}_{tT}^T)^2 + \int_0^t (1-a^T)T\psi^T(T(t-s))\left(\frac{\mu^T}{\overline{\mu}^T} +\; \frac{\alpha}{\overline{\mu}^T}(\mathbf{Z}_{sT}^T)^2\right)ds. \end{equation}
Recall that $\am_t^{*T} \coloneqq \sqrt{\frac{1-a^T}{T\overline{\mu}^T}} \am_{tT}^T$. Let
\begin{align*}
    \Z^{*T}_t \coloneqq \frac{\Z_{tT}^T }{\sqrt{\overline{\mu}^T}} = \int_0^t k(t-s) d\am^{*T}_s,% \label{not6}
\end{align*}
where the last equality follows from assumption \ref{assumpm3} (iv). We can thus rewrite (\ref{e17}) as
\begin{alignat}{1}
  \al_{tT}^{*T}
&\; =\;\frac{(1-a^T)\mu^T}{\overline{\mu}^T} + (1-a^T) \alpha (\mathbf{Z}_{t}^{*T})^2 + \int_0^t f^T(t-s)\left(\frac{\mu^T}{\overline{\mu}^T} +\; \frac{\alpha}{\overline{\mu}^T}(\mathbf{Z}_{sT}^T)^2\right)ds \nonumber\\ 
  &\; +\; \int_0^t \frac{f^T(t-s)}{\sqrt{\overline{\mu}^T(1-a^T)T}}\;d\am_{s}^{*T}, \label{e18} 
\end{alignat}
where $f^T(s) \coloneqq (1-a^T)T\psi^T (Ts).$
Note that $\Lam_t^T = \int_0^t \al_{sT}^{*T}ds$. Let $F^T(t) \coloneqq \int_0^t f^T(s) \; ds$. Substituting from (\ref{e18}) and interchanging the order of the integrals we obtain
 \begin{alignat}{1}\label{nu18}
  \Lam_t^T  &=\;\frac{(1-a^T)\mu^Tt}{\overline{\mu}^T} +  \int_0^t (1-a^T) \alpha (\mathbf{Z}_{s}^{*T})^2 ds +\;\int_0^t \frac{F^T(t-s)}{\sqrt{\overline{\mu}^T(1-a^T)T}}\;d\am_{s}^{*T}\\ \nonumber 
  &\; +\; \int_0^t F^T(t-s)\left(\frac{\mu^T}{\overline{\mu}^T} + \alpha(\mathbf{Z}_{s}^{*T})^2\right)ds.
\end{alignat}
Similarly, for the process $\overline{\al}_{t}^{*T}=\;\frac{(1-a^T)}{\overline{\mu}^T}\olam_t^T$ we have
\begin{alignat*}{1}%\label{nu9}
  \olam_{tT}^{*T} &=\;(1-a^T) +  \frac{1-a^T}{\overline{\mu}^T} \left(\frac{\alpha_1+\alpha_2}{2}\right) (\mathbf{Z}_{tT}^T)^2  +
  \; +\; \int_0^t \frac{(1-a^T)T\overline{\psi}^T(T(t-s))}{\sqrt{\overline{\mu}^T(1-a^T)T}}\;d\overline{\am}_{Ts}^{*T}\\ \nonumber
  & + \; \int_0^t (1-a^T)T\overline{\psi}^T(T(t-s))\left(1 + \frac{(\alpha_1+\alpha_2)}{2\overline{\mu}^T}(\mathbf{Z}_{sT}^T)^2\right)ds,
\end{alignat*}
where $\overline{\psi}^T \coloneqq \sum\limits_{i\geq0}(\beta^T \overline{\phi}^T)^{\circledast i}$. Further, if we define 
\begin{equation*}
 \overline{F}^T(t) \coloneqq \int_0^t \overline{f}^T(s) ds , \hbox{ with } \overline{f}^T(s) \coloneqq (1-a^T)T\overline{\psi}^T (Ts), 
\end{equation*}
we obtain
 \begin{alignat}{1}\label{nu19}
  \oLam_t^T = \int_0^t \olam_{sT}^{*T} &=\;(1-a^T)t +  \int_0^t(1-a^T)\left(\frac{\alpha_1+\alpha_2}{2}\right) (\mathbf{Z}_{s}^{*T})^2 ds 
  \; +\; \int_0^t \frac{\overline{F}^T(t-s)}{\sqrt{\overline{\mu}^T(1-a^T)T}}\;d\overline{\am}_{Ts}^{*T}\\ \nonumber
  & + \; \int_0^t \overline{F}^T(t-s)\left(1 + \frac{(\alpha_1+\alpha_2)}{2}(\mathbf{Z}_{s}^{*T})^2\right)ds.
\end{alignat}
% #------------
From (\ref{nu18}), (\ref{nu19}) we see that to find the dynamics of the limit points for $(\Lam_t^T, \oLam_t^T)_{T > 0}$, we need to find the limit points of  
$(\am_t^{*T}, \overline{\am}_t^{*T}, \mathbf{Z}_{s}^{*T})_{T > 0}$.

Since $\abs{\Delta\am^{*T}}, \abs{\Delta\overline{\am}^{*T}} \leq \frac{2(1-a^T)}{\overline{\mu}^T\sqrt{T}}$, we have from Corollary IX-1.19 in \cite{jacod} that 
$$[\am^{*T}]_t = [\overline{\am}^{*T}]_t  =  \aox_{t}^T.$$
\comment{ , where \begin{align*}%\label{overlinex}
\aox_{t}^T \coloneqq  \xt{1}{t} + \xt{2}{t} = \frac{(1-a^T)}{\overline{\mu}^T}\frac{(\nt{1}{tT} + \nt{2}{tT})}{T}. 
\end{align*}
Thus, by using similar arguments as in section \ref{tightness_nu} we will have that $\aox_{t}^T$ is C-tight. Further, on similar lines of Lemma \ref{l10} we will have $\left(\aox_{t}^T - \oLam_t^{T}\right)_{T > 0}$ converges uniformly to $0$ in probability on $[0,1]$.
Consequently, by corollary VI-6.29 in \cite{jacod} we will have $[M^*]_t = \overline{X}_t$. }

Further, $M^*_t, \overline{M}^*_t$ are continuous as they are limits of  C-tight processes $\{(\am_t^{*T}, \overline{\am}_t^{*T}): t \in [0,1]\}_{T > 0}$. Consequently, we have
\begin{align*}%\label{24.1_nu}
[M^*]_t = \langle M^* \rangle _t = [\overline{M}^*]_t = \langle \overline{M}^* \rangle _t = \overline{X}_t .
\end{align*}
Moreover, $\E[\aox_{1}^T]$ is uniformly bounded as  $\E[\xt{1}{1}]$ and $\E[\xt{2}{1}]$ are uniformly bounded in T by \eref{expect_unstable}. By Fatou's Lemma we have that $\overline{X}_1$ is integrable. Hence, we have $\E([M^*]_t) = \E([\overline{M}^*]_t) =\E[\overline{X}_t] \leq \E[\overline{X}_1]$ is bounded, and thus  $M^*, \overline{M}^*$ are true martingales.

We are now ready to derive the limit of each term on the right in \eref{nu18}. The first term on the RHS  of \eref{nu18} converges to $0$ as $T \to \infty$ as $(1-a^T) \to 0$ and $\frac{\mu^T}{\overline{\mu}^T} \to \frac{\mu^*}{\overline{\mu}^*}$ (from assumption \ref{assumpm3} (iii)).

Consider the second term on RHS of \eref{nu18}. Recall that $\Z^{*T}_t \coloneqq \frac{\Z_{tT}^T }{\sqrt{\overline{\mu}}} = \int_0^t k(t-s) d\am^{*T}_s$. Using integration by parts we obtain
$$\Z^{*T}_t = \int_0^t k(t-s) d\am^{*T}_s = k(0)\am^{*T}_t + \int_0^t k'(t-s)\am_s^{*T}ds,$$
where $k'$ is the derivative of $k$. By assumption \ref{assumpm3} (iv) the subsequential limit $Z^*_t$ of $\Z^{*T}_t$ exists and satisfies
$$Z^*_t = \int_0^t k(t-s) dM^{*}_s = k(0)M^{*}_t + \int_0^t k'(t-s)M_s^{*}ds,$$
which is continuous since $k'$, $M^*$ are continuous. 
Thus, the second term on the right in \eref{nu18} converges to $0$.
 
Using Lemma 4.3 in \cite{rough_fractional} we have for the fourth term on the right in \eref{nu18} 
 \begin{align}
  F^T(t) \coloneqq T\int_0^tT(1-a^T)\psi^T(Ts)ds \to \frac{1}{c_1}F^{\tilde{\alpha}, \sigma}(t),\qquad \mbox{ as } T \to \infty,\label{f1}
\end{align}
where the convergence is uniform in $t$. Here,  $F^{\tilde{\alpha}, \sigma}(t) \coloneqq \int_0^t f^{\tilde{\alpha}, \sigma}(s)ds,$ and $f^{\tilde{\alpha}, \sigma}(s) \coloneqq \sigma s^{\tilde{\alpha} -1} \sum_{n=0}^{\infty}\frac{(-\sigma s^{\tilde{\alpha}})^n}{\Gamma(\tilde{\alpha} n + \tilde{\alpha})}$ with the summation term in $f^{\tilde{\alpha}, \sigma}(s)$ being the Mittag-Lefflar function as defined in \eref{mittag}.   
It follows that
 $$\int_0^t F^T(t-s)(\Z^{*T}_s)^2ds \xrightarrow{T \to \infty} \int_0^t \frac{1}{c_1}F^{\tilde{\alpha}, \sigma}(t-s)(Z^{*}_s)^2ds.$$
 Consequently, the fourth term on the right in \eref{nu18} converges to $\int_0^t \frac{1}{c_1}F^{\tilde{\alpha}, \sigma}(t-s)\left(\frac{\mu^*}{\overline{\mu}^*} + \alpha(Z_{s}^{*})^2\right)ds.$

It remains to analyze the third term on the right in \eref{nu18}. Applying integration by parts we get
 \begin{alignat}{1} \label{eq87}
     \int_0^t \frac{F^T(t-s)}{\sqrt{\overline{\mu}^T(1-a^T)T}}\;d\am_{s}^{*T} 
%     & = \left[\frac{F^T(t-s)}{\sqrt{\overline{\mu}^T(1-a^T)T}}\;\am_{s}^{*T}\right]_0^t - \int_0^t \frac{-f^T(t-s)}{\sqrt{\overline{\mu}^T(1-a^T)T}}\;\am_{s}^{*T}ds\nonumber\\
  & = \int_0^t \frac{f^T(t-s)}{\sqrt{\overline{\mu}^T(1-a^T)T}}\;\am_{s}^{*T}ds.
\end{alignat}
To obtain the limit for the expression on the right in (\ref{eq87}) we write
 \begin{alignat}{1}\label{eq88}
 \int_0^t f^T(t-s)\am_{s}^{*T}ds = &\int_0^t \frac{1}{c_1}f^{\tilde{\alpha}, \sigma}(t-s)M_{s}^{*}ds
  + \int_0^t \frac{1}{c_1}f^{\tilde{\alpha}, \sigma}(t-s)(\am_{s}^{*T}-M_{s}^{*})ds \nonumber\\
  &  + \int_0^t \left(f^T(t-s)-\frac{1}{c_1}f^{\tilde{\alpha}, \sigma}(t-s)\right)\am_{s}^{*T}ds. 
 \end{alignat}
 The second term on the right in the above equation converges uniformly to $0$ almost surely by the remark above \eref{e16}. Applying integration by parts to the third term on the right in \eref{eq88}, we get
 \begin{equation*}%\label{eq89}
 \int_0^t \left(f^T(t-s)-\frac{1}{c_1}f^{\tilde{\alpha}, \sigma}(t-s)\right)\am_{s}^{*T}ds = \int_0^t \left(F^T(t-s)-\frac{1}{c_1}F^{\tilde{\alpha}, \sigma}(t-s)\right)d\am_{s}^{*T}.
 \end{equation*}
By the Burkholder-Davis-Gundy inequality we obtain ($C_i$ below are positive constants)
%
%  \begin{alignat}{1}
%  \E\left[\lim\sup_{t \in [0,1]}\left( \int_0^t\left(F^T(t-s)-\frac{1}{c_1}F^{\tilde{\alpha}, \sigma}(t-s)\right)d\am_{s}^{*T} \right)^2\right] \nonumber\\
% & \hspace{-100pt} \leq C_1 \;\E\left[\lim\sup_{t \in [0,1]} \int_0^t\left(F^T(t-s)-\frac{1}{c_1}F^{\tilde{\alpha}, \sigma}(t-s)\right)^2\;d\aox_{s}^{T} \right] \nonumber\\
%  & \hspace{-100pt} \leq C_2\lim\sup_{t \in [0,1]} \int_0^t\left(F^T(t-s)-\frac{1}{c_1}F^{\tilde{\alpha}, \sigma}(t-s)\right)^2 \frac{1-a^T}{\overline{\mu}^T}\E[\olam_{Ts}^T]\;ds \nonumber\\
%  & \hspace{-100pt} \leq C_3\lim\sup_{t \in [0,1]} \int_0^t\left(F^T(t-s)-\frac{1}{c_1}F^{\tilde{\alpha}, \sigma}(t-s)\right)^2\;ds \rightarrow   0, \label{eq90}
%  \end{alignat}
%
\begin{alignat*}{1}
 \E\left[\sup_{t \in [0,1]}\left( \int_0^t\left(F^T(t-s)-\frac{1}{c_1}F^{\tilde{\alpha}, \sigma}(t-s)\right)d\am_{s}^{*T} \right)^2\right] \nonumber\\
& \hspace{-100pt} \leq C_1 \;\E\left[\int_0^t\left(F^T(t-s)-\frac{1}{c_1}F^{\tilde{\alpha}, \sigma}(t-s)\right)^2\;d\aox_{s}^{T} \right] \nonumber\\
 & \hspace{-100pt} \leq C_2 \int_0^t\left(F^T(t-s)-\frac{1}{c_1}F^{\tilde{\alpha}, \sigma}(t-s)\right)^2 \frac{1-a^T}{\overline{\mu}^T}\E[\olam_{Ts}^T]\;ds \nonumber\\
 & \hspace{-100pt} \leq C_3\int_0^t\left(F^T(t-s)-\frac{1}{c_1}F^{\tilde{\alpha}, \sigma}(t-s)\right)^2\;ds \rightarrow   0, %\label{eq90}
 \end{alignat*}

by \eref{f1}. It follows that the third term on the right in \eref{nu18} converges to $\int_0^t \frac{1}{c_1}f^{\tilde{\alpha}, \sigma}(t-s)\frac{1}{\sqrt{\sigma\overline{\mu}^*}}\;M^*_{s}\;ds$.

This completes the proof for the 
dynamics of $\ax$ in the Theorem \ref{theorem3}. The proof for the dynamics of $\aox$ in the Theorem \ref{theorem3} follows from \eref{nu19} using similar arguments and we shall omit it. 
%\end{proof}

\subsection{Regularity property for $\tilde{\alpha} \in (1/2,1)$}

Regularity follows along the same lines as in section 4.3 and 4.4 of \cite{rough_fractional} and so we only provide a sketch. Recall from \eref{X_t} and  \eref{X_bar_t} that

  $$X_t = \;\int_0^t \frac{1}{c_1}f^{\tilde{\alpha}, \sigma}(t-s)\frac{1}{\sqrt{\sigma\overline{\mu}^*}}\;M^*_{s}\;ds  + \int_0^t  \frac{1}{c_1}F^{\tilde{\alpha}, \sigma}(t-s)\left(\frac{\mu^*}{\overline{\mu}^*} + \alpha(Z^*_{s})^2\right)ds,$$
  and 
 $$\overline{X}_t = \;\int_0^t \frac{1}{c_2}f^{\tilde{\alpha}, \sigma}(t-s)\frac{1}{\sqrt{\sigma\overline{\mu}^*}}\;\overline{M}^*_{s}\;ds  + \int_0^t  \frac{1}{c_2}F^{\tilde{\alpha}, \sigma}(t-s)\left(1 + \frac{(\alpha_1+\alpha_2)}{2}(Z^*_{s})^2\right)ds.$$
  Note that $\int_0^t Z_s^* ds$ is continuously differentiable as $Z^*$ is continuous. Using arguments similar to those in Sections 4.3 and 4.4 of \cite{rough_fractional},
  we conclude that $X$ and $\overline{X}$ are almost surely differentiable with respective derivatives $V$ and $\overline{V}$ satisfying
 $$V_t = \;\int_0^t \frac{1}{c_1}f^{\tilde{\alpha}, \sigma}(t-s)\left[\frac{1}{\sqrt{\sigma\overline{\mu}^*}}\;dM^*_{s}  +\left(\frac{\mu^*}{\overline{\mu}^*} + \alpha(Z^*_{s})^2\right)ds\right],$$
and
$$\overline{V}_t = \;\int_0^t \frac{1}{c_2}f^{\tilde{\alpha}, \sigma}(t-s)\left[\frac{1}{\sqrt{\sigma\overline{\mu}^*}}\;d\overline{M}^*_{s}  +\left(1 + \frac{(\alpha_1+\alpha_2)}{2}(Z^*_{s})^2\right)ds\right].$$
Moreover, by theorem V-3.9 in \cite{revuz}, there exists two Brownian motions $B^1$ and $B^2$ such that
    $$M^*_t =\int_0^t \sqrt{\overline{V}_s}\;dB^1_s, \qquad \overline{M}^*_t =\int_0^t \sqrt{\overline{V}_s}\;dB^2_s$$ 
as $[M^*]_t = [\overline{M}^*]_t = \overline{X}_t$ and  $\overline{X}$ is almost surely differentiable with derivative $\overline{V}_t$. Consequently,
$$Z_t^* = \int_0^t k(t-s)\, dM^{*}_s = \int_0^t k(t-s)\sqrt{\overline{V}_s}\;dB^1_s,$$
from which it follows (see section 4.3 and section 4.4 of \cite{rough_fractional}) that $V_t$ and $\overline{V}_t$ are almost surely $\alpha-\frac{1}{2}-\epsilon$ H\H{o}lder regular. 

\subsection{Correlation between the Brownian motions}

In order to compute $\langle B^1,B^2\rangle_t$ observe that
$$\am_{t}^{*T} = \am_{t}^{*(1,T)}-\am_{t}^{*(2,T)} \hbox{ and } \overline{\am}_t^{*T} = \am_{t}^{*(1,T)}+\am_{t}^{*(2,T)}$$ 
where $\am_{t}^{*(i,T)} = \sqrt{\frac{1-a^T}{T\overline{\mu}^T}}\am_{tT}^{(i,T)}$ for $i \in \{1,2\}.$ Thus
$\am_{t}^{*T}\overline{\am}_t^{*T} = (\am_{t}^{*(1,T)})^2-(\am_{t}^{*(2,T)})^2$
from which it follows that $\am_{t}^{*T}\overline{\am}_t^{*T} - \Lambda_t^T$ is a martingale and hence $\langle \am_{t}^{*T},\overline{\am}_t^{*T}\rangle_t = \Lambda_t^T$. Moreover, by Theorem \ref{theorem3} we have 
$$\am_{t}^{*T} \rightarrow M^*_t =\int_0^t \sqrt{\overline{V}_s}\;dB^1_s\; , \; \overline{\am}_t^{*T} \rightarrow \overline{M}^*_t =\int_0^t \sqrt{\overline{V}_s}\;dB^2_s\; \hbox{ and }\;  \Lambda_t^T \rightarrow \int_0^t V_s \;ds.$$   
It follows that
$$\langle M^*, \overline{M}^*\rangle_t = \int_0^t \overline{V}_s \;d\langle B^1,B^2\rangle_s \;\;\hbox{ and } \;\; \langle M^*, \overline{M}^*\rangle_t = \int_0^t V_s \;ds.$$ 
Comparing the two expressions for $\langle M^*, \overline{M}^*\rangle_t$ we conclude that
$$\langle B^1,B^2\rangle_t = \frac{V_t}{\overline{V_t}}.$$
This completes the proof of Theorem \ref{theorem3}. \qed

%---------------------------------------------
\section{Appendix: Uniqueness of limiting process for the Stable Regime} 
%label{appen}
\label{uniq_m1}
%----------------------------------------------------------------------
% \addtocontents{toc}{\protect\setcounter{tocdepth}{1}}
%\subsection{Uniqueness of limiting process for the Stable Regime} 
% \addtocontents{toc}{\protect\setcounter{tocdepth}{2}}
We will prove the uniqueness of the limit in Theorem \ref{theorem2} using the Banach fixed point theorem for operators (see \cite{picard}, page-4). We will prove uniqueness for $V_t$. Uniqueness for $\overline{V}_t$ can be proved similarly.
% \addtocontents{toc}{\protect\setcounter{tocdepth}{1}}
\subsection{Uniqueness for $V_t$:}
% \addtocontents{toc}{\protect\setcounter{tocdepth}{2}}
%
\begin{align}
 V_t = \mu + \int_0^t \phi(t-s)V_sds + \alpha
 \left(\int_0^t k(t-s)V_sds\right)^2, \qquad t \in [0,1].\label{49}
\end{align}

We define an operator $T: S \rightarrow S$ as
$$T[g](t) = \mu+ \int_0^t \phi(t-s)g(s)ds + \alpha
 \left(\int_0^t k(t-s)g(s)ds\right)^2,$$
 where $S = \{g \in \mathcal{C}[0,1] | \; \norm{g - v_0}_{sup} \leq \eta\}, v_0 = V(0) = \mu \text{ and } \eta$ to be determined later. We wish to choose $\eta$ appropriately so that T satisfies the following two conditions required for applying the Banach fixed point theorem: 
 %In order to use Banach fixed point theorem and conclude the operator $T$ has a fixed point in set $S$ (page-4 \cite{picard}) we need to show:
 %
\begin{enumerate}
    \item $T$ maps $S$ to $S$,
    \item $T$ is a contraction.
\end{enumerate}

\textbf{Step-1:} To find conditions on $\norm{\phi}_1, \norm{k}_2, \text{ and } \eta$ so that $T$ maps $S$ to $S$. For any $g \in S$ we have
\begin{alignat*}{2}
\abs{T[g](t)-v_0(t)} &= \abs{\int_0^t \phi(t-s)g(s)ds + \alpha
 \left(\int_0^t k(t-s)g(s)ds\right)^2}\\
 & \leq \norm{g}_{sup}\int_0^1 \phi(t-s)ds + \alpha \norm{g}^2_{sup}
\left(\int_0^t k(t-s)ds\right)^2\\
 & \leq   \norm{g}_{sup} \norm{\phi}_1 +\alpha \norm{g}^2_{sup}\norm{k}_2^2 
 \end{alignat*}
Thus we have $\norm{T[g]-v_0}_{sup} \leq  \norm{g}_{sup} \norm{\phi}_1 +\alpha \norm{g}^2_{sup}\norm{k}_2^2$.
%
%Now, for $T[g] \in S$ we want that $\norm{T[g]-v_0}_{sup} \leq \eta$, which implies $\norm{g}_{sup}\norm{\phi}_1 +\alpha \norm{g}^2_{sup}\norm{k}_2 \leq \eta$. 
%
As $g \in S$ we have $\norm{g}_{sup} \leq \eta +\abs{\mu}$ and thus $T[g] \in S$ provided
\begin{align}
    (\eta +\abs{\mu})\norm{\phi}_1 +\alpha (\eta +\abs{\mu})^2\norm{k^2}_1 \leq \eta. \label{50}
\end{align}

\textbf{Step-2:} To find conditions on $\norm{\phi}_1, \norm{k}_2, \text{ and } \eta$ so that $T$ is a contraction. For $w,u \in S$ we have
\begin{alignat*}{2}
\abs{T[w](t)- T[u](t)} &\leq \abs{\int_0^t \phi(t-s)(w(s)-u(s))ds} + \alpha\left(\left(\int_0^tk(t-s)w(s)ds\right)^2 - \left(\int_0^tk(t-s)u(s)ds\right)^2\right)\\
&\leq \norm{w-u}_{sup}\norm{\phi}_1 + \alpha\left(\abs{\int_0^tk(t-s)(w(s)-u(s))ds}\right) \left(\abs{\int_0^tk(t-s)(w(s)+u(s))ds}\right)\\
&\leq \norm{w-u}_{sup}\left(\norm{\phi}_1 + \alpha\norm{k}_2^2\norm{u + w}_{sup}\right) \\
&\leq \norm{w-u}_{sup}\left(\norm{\phi}_1 + 2\alpha\norm{k}_2^2\left(\eta + \abs{\mu}\right)\right)
\end{alignat*}%
%
%Hence, we have $\norm{T[w](t)- T[u](t)}_{sup} \leq \norm{w-u}_{sup}\left(\norm{\phi}_1 + 2\alpha\norm{k^2}_1\left(\eta + \abs{\mu}\right)\right)$.
Thus $T$ is a contraction if
\begin{align}
   \norm{\phi}_1 + 2\alpha\norm{k}_2^2\left(\eta + \abs{\mu}\right) < 1.\label{51}
\end{align}
It follows from (\ref{50}) and (\ref{51}) that for $T$ to satisfy the conditions of Banach fixed point theorem it is enough to have 
 \begin{align}\label{e1n}
     \norm{\phi}_1 + 2\alpha\norm{k}_2^2\left(\eta + \abs{\mu}\right) < \frac{\eta}{\eta+\abs{\mu}}.
 \end{align}
 This proves the uniqueness for the process $V_t$.
 
 \comment{On using solution for quadratic equations with inequality it can be seen that for $\eta \in [a,b]$ and for $\norm{k^2}_1$ sufficiently small we will have unique solution for equation(\ref{49}). Here $a = \frac{1-\norm{\phi}_1-\sqrt{\left(\norm{\phi}_1-1\right)^2-8\alpha\norm{k^2}_1\abs{\mu}}}{4\norm{k^2}_1}$ and $b = \frac{1-\norm{\phi}_1+\sqrt{\left(\norm{\phi}_1-1\right)^2-8\alpha\norm{k^2}_1\abs{\mu}}}{4\norm{k^2}_1}$. 
 Thus, once we have a unique solution for small $\eta$, the same will work for a large value of $\alpha$. 
 
 Thus, equation (\ref{49}) has a unique solution in $\left[\abs{\mu}-\eta,\abs{\mu}+\eta\right]$ for $\eta \in [a,b]$, for all $t \in [0,1]$ such that $\left(\norm{\phi}_1-1\right)^2-8\alpha\norm{k^2}_1\abs{\mu} > 0.$}
 
 % \addtocontents{toc}{\protect\setcounter{tocdepth}{1}}
\subsection{Uniqueness for $\overline{V}_t$:}
% \addtocontents{toc}{\protect\setcounter{tocdepth}{2}}

 \begin{align*}
 \overline{V}_t = \overline{\mu} + \int_0^t \overline{\phi}(t-s)\overline{V}_sds + \frac{(\alpha_1+\alpha_2)}{2}
 \left(\int_0^t k(t-s)\overline{V}_sds\right)^2, \text{for } t \in [0,1].%\label{59}
\end{align*}
%
%On similar lines as for uniqueness for process $V_t$, we will have 
Uniqueness for process $\overline{V}_t$ follows along similar lines as above if there exists a positive constant $\overline{\eta}$ such that
\begin{align*}%\label{e2n}
     \norm{\phi}_1 + (\alpha_1+\alpha_2)\norm{k^2}_1\left(\overline{\eta} + \overline{\mu}\right) < \frac{\overline{\eta}}{\overline{\eta}+\overline{\mu}}.
 \end{align*}

\comment{$\overline{V}_t$  in $\left[\overline{\mu}-\eta,\overline{\mu}+\eta\right]$ for $\eta \in [a,b]$, for all $t \in [0,1]$ such that $\left(\norm{\phi}_1-1\right)^2-8\frac{(\alpha_1+\alpha_2)}{2}\norm{k^2}_1\overline{\mu} > 0.$}

%#++++++++++++++++++++++++++++++++++++++++++++++++++++++++++

%============= REFERENCES +++++++++========%

\addcontentsline{toc}{section}{References} \label{bib}


\begin{thebibliography}{}

\bibitem{bacry} E. Bacry, S. Delattre, M. Hoffmann, J.F. Muzy (2013). \href{https://doi.org/10.1016/j.spa.2013.04.007}{Some limit theorems for Hawkes processes and application to financial statistics}. Stochastic Processes and their Applications, 123(7), 2475–2499.

\bibitem{finan} Emmanuel Bacry, Iacopo Mastromatteo, Jean-Francois Muzy (2015). \href{https://arxiv.org/abs/1502.04592}{Hawkes processes in finance}. Market Microstructure and Liquidity, 1(1), 1550005.

\bibitem{ahawkes} A. G. Hawkes (1971). \href{https://doi.org/10.2307/2334319}{Spectra of Some Self-Exciting and Mutually Exciting Point Processes}. Biometrika, 58(1), 83–90.

\bibitem{close_to_1} Stephen J. Hardiman, Nicolas Bercot, Jean-Philippe Bouchaud (2013). \href{https://arxiv.org/abs/1302.1405}{Critical reflexivity in financial markets: A Hawkes process analysis}. Physica A: Statistical Mechanics and its Applications, 392(17), 3943–3954.

\bibitem{rough_heston} Omar El Euch, Mathieu Rosenbaum (2018). \href{https://www.jstor.org/stable/26542503}{Perfect hedging in rough Heston models}. The Annals of Applied Probability, 28(6), 3813–3856.

\bibitem{omar} Omar El Euch, Masaaki Fukasawa, Mathieu Rosenbaum (2016). \href{https://arxiv.org/abs/1609.05177}{The microstructural foundations of leverage effect and rough volatility}. Quantitative Finance, 18(9), 1391–1402.

\bibitem{mar} Marian-Andrei Rizoiu, Young Lee, Swapnil Mishra (2017). \href{https://arxiv.org/abs/1708.06401}{Tutorial on Hawkes processes for Events in Social Media}. IEEE International Conference on Data Mining Workshops (ICDMW).

\bibitem{abi} Eduardo Abi Jaber, Christa Cuchiero, Martin Larsson, Sergio Pulido (2019). \href{https://arxiv.org/abs/1909.01166}{A weak solution theory for stochastic Volterra equations of convolution type}. Finance and Stochastics, 24, 379–422.

\bibitem{revuz} D. Revuz and M. Yor. Continuous Martingales and Brownian Motion, third edition. Springer, 1999.

\bibitem{zhang} Xicheng Zhang (2008). \href{https://arxiv.org/abs/0812.0834}{Stochastic Volterra Equations in Banach Spaces and Stochastic Partial Differential Equations}. Journal of Functional Analysis, 255(7), 1878–1913.

\bibitem{martingale} Steven P. Lalley (2016). \href{https://galton.uchicago.edu/~lalley/Courses/385/ContinuousMG1.pdf}{Continuous Martingales}. University of Chicago Lecture Notes.

\bibitem{aditi} Aditi Dandapani, Paul Jusselin, Mathieu Rosenbaum (2021). \href{https://arxiv.org/abs/1907.06151}{From quadratic Hawkes processes to super-Heston rough volatility models with Zumbach effect}. Quantitative Finance, 21(4), 601–618.

\bibitem{benoit} Benoit Mandelbrot (1963). \href{http://www.jstor.org/stable/2350970}{The Variation of Certain Speculative Prices}. The Journal of Business, 36(4), 394–419.

\bibitem{blanc} Pierre Blanc, Jonathan Donier, Jean-Philippe Bouchaud (2015). \href{http://dx.doi.org/10.2139/ssrn.2665669}{Quadratic Hawkes Processes for Financial Prices}. Quantitative Finance, 16(7), 1–25.

\bibitem{rough_fractional} Thibault Jaisson, Mathieu Rosenbaum (2015). \href{https://arxiv.org/abs/1504.03100}{Rough fractional diffusions as scaling limits of nearly unstable heavy-tailed Hawkes processes}. Annals of Applied Probability, 26(6), 3547–3581.

\bibitem{flan} Franco Flandoli, Dariusz Gatarek (1995). \href{https://link.springer.com/article/10.1007/BF01192467}{Martingale and stationary solutions for stochastic Navier-Stokes equations}. Probability Theory and Related Fields, 102(3), 367–391.

\bibitem{jacod} J. Jacod, A. N. Shiryaev (1987). Limit Theorems for Stochastic Processes. Grundlehren der Mathematischen Wissenschaften [Fundamental Principles of Mathematical Sciences], Springer.

\bibitem{stable_hawkes} Cecilia Aubrun, Michael Benzaquen, Jean-Philippe Bouchaud (2021). \href{https://arxiv.org/abs/2112.14161}{On Hawkes Processes with Infinite Mean Intensity}. Journal of Statistical Mechanics: Theory and Experiment.

\bibitem{nonlinear_hawkes} Pierre Bremaud, Laurent Massoulie (1996). \href{https://www.jstor.org/stable/2244985}{Stability of Nonlinear Hawkes Processes}. The Annals of Probability, 24(3), 1563–1588.

\bibitem{kleb} Fima C. Klebaner (2005). Introduction to Stochastic Calculus with Applications, Second Edition. Imperial College Press.

\bibitem{bilin} Patrick Billingsley. Probability and Measure, third edition. Wiley, 1995.

\bibitem{daley} D. J. Daley, D. Vere-Jones (2003). An Introduction to the Theory of Point Processes: Volume I: Elementary Theory and Methods, Second Edition. Springer.

\bibitem{power} Thibault Jaisson, Mathieu Rosenbaum (2018). \href{https://arxiv.org/pdf/1805.07134.pdf}{No-arbitrage implies power-law market impact and rough volatility}. Finance and Stochastics, 22(2), 399–426.

\bibitem{unstable} Thibault Jaisson, Mathieu Rosenbaum (2015). \href{http://www.jstor.org/stable/24519929}{Limit theorems for nearly unstable Hawkes processes}. The Annals of Applied Probability, 25(2), 600–631.

\bibitem{rough} Jim Gatheral, Thibault Jaisson, Mathieu Rosenbaum (2014). \href{https://arxiv.org/abs/1410.3394}{Volatility is rough}. Quantitative Finance, 16(8), 1109–1124.

\bibitem{time_reversal} Gilles Zumbach (2007). \href{https://arxiv.org/abs/0708.4022}{Time reversal invariance in finance}. Quantitative Finance, 9(6), 611–620.

\bibitem{picard} Denise Gutermuth. Picard’s Existence and Uniqueness Theorem. University of Pittsburgh Lecture Notes.

\end{thebibliography}
\end{document}